\documentclass{llncs}
\usepackage{makeidx}  
\usepackage{amssymb}
\usepackage{amsmath}
\usepackage{wasysym}
\usepackage{textcomp}
\usepackage{enumerate}
\usepackage{comment}
\newcommand{\dollar}[0]{\$}

\begin{document}

\title{HOMING VECTOR AUTOMATA} 


\author{ \"{O}zlem Salehi \thanks{The first author is partially
supported by T\"{U}B\.{I}TAK (Scientific and Technological Research Council of Turkey).}
 \and Ahmet Celal Cem Say }
\authorrunning{"{O}zlem Salehi, Ahmet Celal Cem Say}

\institute{Bo\v{g}azi\c{c}i University, Department of Computer Engineering, Bebek 34342 Istanbul, Turkey\\
\email{ozlem.salehi@boun.edu.tr},
\email{say@boun.edu.tr},
}

\maketitle              

\begin{abstract}
We introduce homing vector automata, which are finite automata augmented by a vector that is multiplied at each step by a matrix determined by the current transition, and have to return the vector to its original setting in order to accept the input. The computational power of the deterministic, nondeterministic and blind versions of these real-time machines are examined and compared to various related types of automata. A generalized version of the Stern-Brocot encoding method, suitable for representing strings on arbitrary alphabets, is also developed.
\end{abstract}

\section{Introduction}

The idea of augmenting the classical finite automaton model with an external storage unit that can hold unlimited amounts of information, yet can be accessed in a limited mode, is a celebrated topic of automata theory, with pushdown automata \cite{Ch62} and counter machines \cite{FMR67} as the most prominent examples. 

Focusing on finite automata equipped with a register containing a singleton, one can list automata with multiplication \cite{ISK76}, automata over groups \cite{MS97} and M-automata \cite{Ka09} among the many such proposed models. In these machines, the register can store rational numbers, elements from a group, or a monoid, and can be modified by multiplication. A computation is deemed successful if the register, which is initialized to the identity element, is equal to the identity element at the end. 
                                      
Generalizing the idea of finite automata equipped with a register, we have previously introduced ``vector automata'' in \cite{SYS13}. A vector automaton is a finite state automaton which is endowed with a vector and which can multiply this vector with an appropriate matrix at each step. The input is read real-time and only one of the entries can be tested for equality to a rational number every step. The machine accepts an input string if the computation ends in an accept state and the test for equivalence succeeds. 

Many important models of probabilistic and quantum computation \cite{Tu69,LR14} can be viewed in terms of vectors being multiplied by matrices. Vector automata are useful for focusing on this matrix multiplication view of programming, abstracting the remaining features of such models away. In order to incorporate the aforementioned notion of the computation being successful if the register/counter returns to its initial value at the end of the computation to this setup,  we propose a new model called ``Homing Vector Automaton'' in this paper. A homing vector automaton can multiply its vector with an appropriate matrix at each step and can check the entire vector for equivalence to the initial value of the vector. The acceptance criterion is ending up in an accept state with the value of the vector being equal to the initial vector. We focus on real-time input throughout the paper. 

We provide an exact characterization of the class of languages recognized by these machines for the case where the alphabet is unary. We define ``blind'' homing vector automata, where the equality test can be performed only at the end of the computation. The blind version of our model can be seen as a generalization of some well known models such as real-time blind multicounter automata \cite{Gr78}.  The nondeterministic version of our model is capable of recognizing some $ \mathsf{NP} $-complete languages. We compare the related language classes recognized by different versions of our model, and show a hierarchy result based on the dimension of the vector when the matrix entries belong to a restricted set. 
A method we use for encoding strings on an alphabet of arbitrary size in a blind homing vector automaton, based on Stern-Brocot trees \cite{St58,Br61}, may be of independent interest. 

\section{Preliminaries}

The following notation will be used throughout the paper: $Q$ is the set of
states, where $q_0 \in Q$ denotes the initial state, $Q_a \subset Q$ denotes the
set of accept states, and $\Sigma$ is the input alphabet. An input string $w$ is placed between two endmarker symbols
on an infinite tape in the form $\cent w\dollar$. We define $ \tilde{\Sigma}=\Sigma \cup \{\cent,\dollar\} $. By $ w^r $, we represent the reverse of the string $ w $. $ w_i $ denotes the $ i $'th symbol of $ w $. 

For a machine model $A$, $\mathfrak{L}(A)$
denotes the class of languages recognized by machine of type $A$.

Throughout the paper we will focus on real-time computation where the input head moves right at each step. All models presented below operate in real-time. We start with multicounter automata.

A \textit{real-time deterministic $k$-counter automaton} (D\textit{k}CA) \cite{FMR68} is a 5-tuple
\[ \mathcal{M}=(Q, \Sigma, \delta, q_0, Q_a). \]

The transition function $\delta$ of $\mathcal{M}$ is specified so that
$\delta(q,\sigma,\theta)=(q',c)$
means that $\mathcal{M}$
moves the head to the next symbol, switches to state $q'$, and updates its counters according to the list of increments represented by $c \in \{-1,0,1\}^k$,
if it reads symbol $\sigma \in \Sigma$, when  in state $q \in Q$, and
with $\theta  \in \{=,\neq\}^k $ describing whether the respective counter values equal zero
or not. At the
beginning of the computation, the tape head is placed on the symbol $\cent$,
and the counters are set to 0. At the end of the computation, that is, after the right endmarker $\dollar$
has been scanned, the input is
accepted if $\mathcal{M}$ is in an accept state.

A \textit{real-time deterministic blind $k$-counter automaton} (D\textit{k}BCA) \cite{Gr78}
$\mathcal{M}$ is a D\textit{k}CA which can
check the value of its counters only at the end of the computation. Formally,
the transition function is now replaced by $\delta(q,\sigma)=(q',c).$ The input
is accepted at the end of the computation if $\mathcal{M}$ enters an accept
state, and all counter values are equal to 0.

A \textit{real-time deterministic vector automaton of dimension $k$} (DVA($ k $)) \cite{SYS13} is a 6-tuple
\[\mathcal{V} =(Q,\Sigma,\delta,q_0,Q_a,v),\]
 where
$v$ is a $k$-dimensional initial row vector, and the
transition function $\delta$ is defined as
\[\delta: Q \times \tilde{\Sigma} \times \Omega \rightarrow Q\times S,\]
such that $S$ is the set of $k \times k$ rational-valued matrices, and $\Omega=\{=,\neq\}$, where $ = $  indicates
equality to 1, and $ \neq $ otherwise.

Specifically, $\delta(q,\sigma,\omega)=(q',M)$ means that when $\mathcal{V}$ is
in state $q$ reading symbol $\sigma \in \tilde{\Sigma}$, and the first entry of its
vector corresponds to $\omega \in \Omega$ (with $\omega$ having the value = if and only if this entry is equal
to 1),
$\mathcal{V}$ moves to state $q'$, multiplying its vector with the matrix $M
\in S$. $\omega$ is taken to be = if the first entry of the vector equals 1,
and $\neq$ otherwise. The string is accepted if
$\mathcal{V}$ enters
an accept state, and the first entry of the vector is 1, after processing the right end-marker symbol
$\dollar$.

\section{Homing Vector Automata}

A \textit{real-time deterministic homing vector automaton} (DHVA(\textit{k})) $ \mathcal{V} $ is a vector automaton which checks the value of the vector for equivalence to the initial vector instead of checking a single entry. Formally, a DHVA(\textit{k})  is a 6-tuple
\[\mathcal{V} =(Q,\Sigma,\delta,q_0,Q_a,v),\]
 where
$v$ is a $k$-dimensional initial row vector, and the
transition function $\delta$ is defined as
\[\delta: Q \times \Sigma  \times \Omega \rightarrow Q\times S,\]
such that $\Omega=\{=,\neq\}$, where $ = $  indicates
equality to initial vector $ v $, and $ \neq $ otherwise and $S$ is the set of $k \times k$ rational-valued matrices. The initial vector is freely chosen by the designer of the automaton.

Specifically, $\delta(q,\sigma,\omega)=(q',M)$ means that when $\mathcal{V}$ is
in state $q$ reading symbol $\sigma \in \Sigma$, and the vector corresponds to $\omega \in \Omega$ (with $\omega$ having the value = if and only if the vector is equal to the initial vector),
$\mathcal{V}$ moves to state $q'$, multiplying its vector with the matrix $M
\in S$ on the right. Thus the vector $ v_i $ at step $ i $ is obtained by multiplying the vector $ v_{i-1} $ at step $ i-1 $ by an appropriate matrix $ M $ so that $ v_i=v_{i-1}M $.
The string is accepted if
$\mathcal{V}$ enters
an accept state, and the vector is equal to the initial vector $ v $ when reading the right end-marker symbol
$\dollar$.

A \textit{real-time deterministic blind homing vector automaton} (DBHVA(\textit{k})) is a
DHVA(\textit{k}) which is not allowed to check the vector until the end of the computation. The transition function $\delta$ is defined as
$$ \delta: Q \times \Sigma \rightarrow Q\times S, $$
with $S$ as defined earlier.
$\delta(q,\sigma)=(q',M)$ means that when $\mathcal{V}$ reads symbol $\sigma \in \Sigma$
in state $q$, it will move to state
$q'$, multiplying the vector  with the matrix $M \in S$. The acceptance condition is the same as for
DHVA($ k $)'s. 

A \textit{real-time nondeterministic homing vector automaton} (NHVA(\textit{k})) is a DHVA(\textit{k}) which has the additional capability of making nondeterministic choices. The transition function $\delta$ is now replaced by 
$$\delta: Q \times
\Sigma \times \Omega \rightarrow \mathbb{P}(Q\times S), $$ 
where $\mathbb{P}(A)$ denotes the
power set of the set $A$. 

A \textit{real-time
nondeterministic blind homing vector automaton} (NBHVA(\textit{k})) is just a NHVA(\textit{k}) which does not check
the vector until the end of the computation. The transition function $ \delta $ is defined as
$$ \delta: Q \times \Sigma \rightarrow \mathbb{P}(Q\times S). $$

\section{Blindness, Tally Languages, and Nondeterminism}

The definition of homing vector automata allows arbitrary rational matrices. In most automaton algorithms in this paper, the entries of the matrices belong to the set $ \{-1,0,1\} $, since this basic set already captures many capabilities of homing vector automata. Let us note that multiplications with matrices whose entries belong to this set can be used to perform additions, subtractions, resets, and swaps between the vector entries. It is possible to recognize some of the languages in the following discussion with homing vector automata of lower dimension when a larger set of matrix entries is allowed. Some related open questions can be found in Section \ref{sec:end}.

We start by comparing the blind and non-blind versions of our model. 

\begin{theorem}\label{thm:blind}
$ \bigcup_k \mathfrak{L} 
\textup{(DBHVA(\textit{k}))} \subsetneq \bigcup_k \mathfrak{L} \textup{(DHVA(\textit{k}))}. $
\end{theorem}
\begin{proof}
It is obvious that any DBHVA($ k $) can be simulated by a DHVA($ k $). We are going to prove that the inclusion is proper by the witness language $ \mathtt{L}=\{a^nb^{a_1}a^{a_2}|n=a_1 \mbox{ or } n=a_1 + a_2\} $. Let us first construct a DHVA(2) $ \mathcal{V} $ recognizing $\mathtt{L} $. The idea is to simulate a counter with the help of the matrices. Starting with the initial vector 
$ 
\left [
\begin{array} {rr}
1&1\\
\end{array}
\right ]
 $, $\mathcal{V} $ multiplies the vector with the matrix $ M_+ $ for each  $ a $ it reads before the $b$'s,  incrementing the first entry of the vector with each such multiplication. After finishing reading the first segment of $ a $'s, $\mathcal{V} $ multiplies the vector with the matrix $ M_- $, decrementing the first entry of the vector for each $ b $.

$$M_{+}=
\left [
\begin{array} {rr}
1&0\\
1&1\\
\end{array}
\right ]~~~~~
M_{-}=
\left [
\begin{array}{rr}
1&0\\
-1&1\\
\end{array}
\right ]
$$ 

At each step, $ \mathcal{V} $ checks the current value of the vector for equality to $ 
\left [
\begin{array} {rr}
1&1\\
\end{array}
\right ]
$. If the equality is detected right after finishing reading the $ b $'s, it is the case that $ n=a_1 $, and $ \mathcal{V} $ multiplies the vector with the identity matrix at each step for the rest of the computation. If that is not the case, $ \mathcal{V} $ continues to multiply the vector with matrix $ M_- $ for each $ a $ after the $b$'s. The value of the vector will be equal to 
$  \left [
 \begin{array} {rr}
 1&1\\
 \end{array}
 \right ] $ at the end of the computation  if and only if $ n=a_1 $ or $ n =a_1+a_2$. 
 
Note that $ \mathtt{L} $ can be also recognized by a DHVA(1) by using the matrices $ M_{+}=2 $ and $ M_{-}=\frac{1}{2} $.
 
Now we are going to show that $ \mathtt{L} $ can not be recognized by any DBHVA($ k $). Suppose for a contradiction that $ \mathtt{L} $ is recognized by some DBHVA($ k $) $ \mathcal{V'} $. After reading a sufficiently long input prefix of the form $ a^n$, the computation of $ \mathcal{V'} $ on a sufficiently long postfix of $b$'s will go through a sequence of states, followed by a state loop. Suppose that $ \mathcal{V'} $ is in the same state after reading two different strings $ a^nb^m $ and $ a^nb^n $, $ m<n $. Now consider the strings $u= a^nb^ma^{n-m} \in \mathtt{L} $ and $ w=a^nb^na^{n-m} \in \mathtt{L} $. After reading any one of these strings, $ \mathcal{V'} $ should be in the same accept state, and the  vector should be at its initial value. Assume that the strings in question are both extended with one more $ a $.  Since the same vector is being multiplied with the same matrix associated with the same state during the processing of that last $ a $, it is not possible for $ \mathcal{V'} $ to give different responses to $ a^nb^na^{n-m+1}$ and $ a^nb^ma^{n-m+1}$. Noting that $ a^nb^na^{n-m+1} \in \mathtt{L}$, whereas $ a^nb^ma^{n-m+1} \notin \mathtt{L}$, we conclude that $ \mathtt{L} $ can not be recognized by any DBHVA($ k $). 
\end{proof}

We can give the following characterization when the alphabet is unary.

\begin{theorem}\label{thm:unary}
For any $k$, all languages over $\Sigma = \{a\} $ accepted by a \textup{DHVA($ k $)} are regular.
\end{theorem}

\begin{proof}
Let $ \mathtt{L} $ be a unary language accepted by a DHVA($k$) $ \mathcal{V} $ and let $ v $ be the initial vector of $ \mathcal{V}$. We are going to construct a DFA recognizing $ \mathtt{L} $ to prove that $ \mathtt{L} $ is regular. We assume that $ \mathtt{L} $ is infinite and make the following observation. Since $ \mathcal{V} $ has finitely many states, at least one of the accept states of $ \mathcal{V} $ will be accepting more than one string. Let $ w_1 $ and $ w_2 $ be the shortest strings accepted by an accept state $ q_a $ with $ |w_1|<|w_2| $. When accepting $ w_1 $ and $ w_2$, $ \mathcal{V} $ is in state $ q_a $ and the value of the vector is equal to $ v $. After reading $ w_2 $, $ \mathcal{V} $ is in the same configuration as it was after reading $ w_1 $ and this configuration will be repeated inside a loop of $|w_2|-|w_1|= p $ steps. Therefore, we can conclude that all strings of the form $ a^{|w_1|+kp} $ for some positive integer $ k $ will be accepted by $ q_a $. 
 
Between consecutive times $ q_a $ accepts a string, some other strings may be accepted by some other accept states. Let $ u $ be a string accepted by $ q_b $ with $ |w_1| < |u| < |w_2| $. Then all strings of the form $ a^{|u|+kp} $ for some positive integer $ k $ will be accepted by $ q_b$ since every time $ \mathcal{V} $  enters the accepting configuration at state $ q_a $, $ \mathcal{U} $ will enter the accepting configuration at state $ q_b $ after $ |u|-|w_1| $ steps. The same reasoning applies to any other accepting configuration inside the loop. 
 
Now, let us construct a DFA $ \mathcal{D} $ accepting $ \mathtt{L} $. $ \mathcal{D} $ has $ |w_1|+1+(p-1) $ states. The first $ |w_1|+1 $ states correspond to the strings of length at most $ |w_1| $ and the state $ q_{|w|} $ is an accept state if $ w \in \mathtt{L}$. $ q_{|w_1|} $ and the next $ p-1 $ states $ q_{l_2},\dots,q_{l_p} $ stand for the configuration loop. States corresponding to accepting configurations inside the loop are labeled as accept states. 

The transitions of the DFA are as follows:
\begin{align*}
 \delta(q_i,a)&=q_{i+1} \mbox{ for }  i=0,\dots,|w_1|-1 \\
 \delta(q_{|w_1|},a)&=q_{l_2} \\
 \delta(q_{l_i},a)&=q_{l_{i+1}}  \mbox{ for }  i=2,\dots,p-1 \\
 \delta(q_{l_p},a)&=q_{|w_1|} \\
\end{align*}

Since $ \mathtt{L} $ can be recognized by a DFA, $  \mathtt{L} $ is regular. We conclude that any unary language accepted by a  \textup{DHVA($ k $)}  is regular.
\end{proof}

In the following theorem, we show that nondeterministic homing vector automata are more powerful than their deterministic versions, both in the blind and nonblind cases.
\begin{theorem}\label{thm:upow}
\begin{enumerate}[i.]
\item $ \bigcup_k \mathfrak{L}\textup{(DBHVA(\textit{k}))} \subsetneq \bigcup_k \mathfrak{L} \textup{(NBHVA(\textit{k}))}$. 
\item $ \bigcup_k \mathfrak{L}\textup{(DHVA(\textit{k}))} \subsetneq \bigcup_k \mathfrak{L} \textup{(NHVA(\textit{k}))}$. 
\end{enumerate}
\end{theorem}
\begin{proof}
\textit{i.} It is obvious that a DBHVA($ k $) can be simulated by a NBHVA($ k $). We are going to show that the inclusion is proper by constructing a NBHVA(3) $ \mathcal{V} $ recognizing the unary nonregular language $ \mathtt{UPOW}=\{a^{n+2^n}|n\geq 1\} $. Starting with the initial vector $  \left [
\begin{array} {rrr}
1& 1 & 1\\
\end{array}
\right ] $, $ \mathcal{V} $ multiplies the vector with matrix $ U_1 $ when reading each $ a $. The idea is to add the first and second entries together repeatedly to obtain powers of 2, so that after reading $ k $ symbols the value of the vector is equal to 
$  \left [
 \begin{array} {rrr}
  2^k&2^k&1\\
  \end{array}
\right ] $. $ \mathcal{V} $ nondeterministically guesses $ n $ and starts decrementing the first entry from that point on by multiplying the vector with the matrix $ U_2 $. At the end of the computation, the value of the vector is equal to 
$  \left [
   \begin{array} {rrr}
   1& 1 & 1\\
   \end{array}
   \right ] $ if and only if the input string is of the form $ a^{n+2^n} $ for some $ n $. 
   $$U_{1}=
   \left [
   \begin{array} {rrr}
   1&1&0\\
   1&1&0\\
   0&0&1\\
   \end{array}
   \right ]~~~~~
   U_{2}=
   \left [
   \begin{array}{rrr}
   1&0&0\\
   0&0&0\\
   -1&1&1\\
   \end{array}
   \right ]
   $$ 
From Theorem $ \ref{thm:unary}$, we know that every unary language recognized by a DHVA($ k $) is regular, concluding that $ \mathtt{UPOW}\notin \bigcup_k \mathfrak{L} \textup{(DBHVA(\textit{k}))}$ .\\

\noindent \textit{ii.} It is obvious that a DHVA($ k $) can be simulated by a NHVA($ k $). The inclusion is proper as we have shown that $ \mathtt{UPOW} $ can be recognized by a NHBVA(3), a feat that is impossible for DHVA($ k $)'s for any $k$.

Let us remark that it is possible to recognize $ \mathtt{UPOW} $ by a NBHVA(2) when the matrix entries are not restricted to the set $ \{-1,0,1\} $.
\end{proof}

In the following theorem, we show that by allowing nondeterminism it is possible to recognize an $\mathsf{NP}$-complete language. $\mathtt{SUBSETSUM}$ is the $\mathsf{NP}$-complete language which is the collection of all strings of the form $t \#  a_1\#...\# a_n\#$, such
that $t$ and the $a_i$'s are numbers in binary notation $(1 \leq i \leq n)$, and there
exists a set $I \subseteq \{1, . . . , n\}$ satisfying $\sum_{i \in
I}a_i=t$, where $n > 0$.  We define $\mathtt{SUBSETSUM}_r=\{ t^r \#  a_1^r\#...\# a_n^r\#\ | \exists I \subseteq \{1, . . . , n\} \mbox{ s.t. } \sum_{i \in
I}a_i=t\}$ in which the binary numbers appear in reverse order. It is obvious that $\mathtt{SUBSETSUM}_r \in \mathsf{NP}$, since $ \mathtt{SUBSETSUM}  \in \mathsf{NP} $. It is possible to reduce $\mathtt{SUBSETSUM}$ to $ \mathtt{SUBSETSUM}_r$ in polynomial time by reversing the binary numbers that appear in the input. Therefore, we can conclude that $ \mathtt{SUBSETSUM}_r$ is $\mathsf{NP}$-complete.   

\begin{theorem}
$ \mathtt{SUBSETSUM}_r \in \mathfrak{L}(\textup{NBHVA(5)}) $.
\end{theorem}
\begin{proof}
We construct a NBHVA(5) $\mathcal{V}$ recognizing 
$\mathtt{SUBSETSUM}_r$. The idea of this construction is to read the binary numbers in the string to entries of the vector, and to nondeterministically select the set of numbers that add up to $ t $.  We let the initial vector equal $ \left[
\begin{array}{rrrrr}
0&0&1&1&1
\end{array} \right ] $. We first encode $ t $ to the first entry of the vector as follows: While scanning the symbols of $t$, $\mathcal{V}$ multiplies the vector with the matrix $M_{T_0}$ (resp. $M_{T_1}$) for each scanned $0$
(resp. $1$). The  powers of 2 required for the encoding are obtained by adding the third and fourth entries, which always contain identical numbers, to each other, creating the effect of multiplication by 2. When $\mathcal{V}$ reads a $\#$, $\mathcal{V}$ multiplies the vector with the matrix $ M_{\#} $ which subtracts the second entry from the first entry and resets the second entry back to 0, and the third and  fourth entries back to 1.   
\[
 M_{T_0}=
\left[
\begin{array}{rrrrr}
1&0&0&0&0 \\
0&1&0&0&0\\
0&0&1&1&0\\
0&0&1&1&0\\
0&0&0&0&1
\end{array}
\right ]~~~~~
M_{T_1}=
\left[
\begin{array}{rrrrr}
1&0&0&0&0 \\
0&1&0&0&0\\
1&0&1&1&0\\
0&0&1&1&0\\
0&0&0&0&1
\end{array}
\right ]~~~~~
 M_{\#}=
 \left[
 \begin{array}{rrrrr}
 1&0&0&0&0 \\
 -1&0&0&0&0\\
 0&0&0&0&0\\
 0&0&0&0&0\\
 0&0&1&1&1
 \end{array}
 \right ] \]
 
In the rest of the computation, $\mathcal{V}$ nondeterministically decides which $a_i$'s to
subtract from the first entry.
Each selected $a_i$ is encoded using the same technique into the second entry of the vector. While scanning the symbols of $a_i$, $\mathcal{V}$ multiplies the vector with the matrix $M_{A_0}$ (resp. $M_{A_1}$) for each scanned $0$
(resp. $1$).

$$
M_{A_0}=
\left[
\begin{array}{rrrrr}
1&0&0&0&0 \\
0&1&0&0&0\\
0&0&1&1&0\\
0&0&1&1&0\\
0&0&0&0&1
\end{array}
\right ]~~~~~
M_{A_1}=
\left[
\begin{array}{rrrrr}
1&0&0&0&0 \\
0&1&0&0&0\\
0&1&1&1&0\\
0&0&1&1&0\\
0&0&0&0&1
\end{array}
\right ]
.$$

  $\mathcal{V}$ chooses another
$a_j$ if it wishes, and the same procedure is applied. At
the end of the input, $\mathcal{V}$ accepts if the vector is equal to 
$  \left[
 \begin{array}{rrrrr}
0&0&1&1&1 \\
 \end{array}
 \right ] $, which requires that the first entry of the vector is equal to 0. This is possible iff there exists a set of $ a_i $'s whose sum add up to $ t $.
\end{proof}

A language $ \mathtt{L} $ is in class $ \mathsf{TISP} $($ t(n),s(n) $) if there is a deterministic Turing machine that decides $ \mathtt{L} $ within $ t(n) $ time and $ s(n) $ space where $ n $ is the length of the input. Since the numbers in the vector can grow by at most a fixed number of bits in each multiplication, a Turing machine simulating a DHVA($ k $) requires only linear space \cite{SYS13}. Since the numbers in the vector can have length $O(n)$, whereas the matrix dimensions and entries are independent of the input length $n$, multiplication of a vector and a matrix requires $ O(n) $ time for each input symbol. We can conclude that $ \bigcup_k \mathfrak{L}$(DHVA($ k $))$ \subseteq  \mathsf{TISP}( n^2,n )$.

\section{Encoding Strings with Homing Vector Automata}

\subsection{Stern-Brocot Encoding}
The Stern-Brocot tree is an infinite complete binary tree whose nodes correspond one-to-one to positive rational numbers \cite{St58,Br61}. Crucially for our purposes, the Stern-Brocot tree provides a basis for representing strings as vectors of integers, as suggested for binary alphabets in \cite{GKP89}. 
When fractions are represented as vectors of dimension 2, where the entries correspond to the denominator and the numerator of the fraction, this encoding can be done easily in homing vector automata, as follows.

The empty string is represented by $ [\begin{array} {rr}
1&1
\end{array}] $. Now suppose that we want to encode a binary string $ w $ of length $ n $. For $ i=1 $ to $ n $, if $ w_i=0 $, we add the value of the first entry to the second one, and if $ w_i=1 $, we add the value of the second entry to the first one, multiplying the vector with the appropriate one of the following matrices $ M_0 $ and $ M_1 $:
$$ 
 M_{0}=
 \left [
 \begin{array}{rr}
 1&1\\
 0&1\\
 \end{array}
 \right ]~~~~~
 M_{1}=
 \left [
 \begin{array}{rr}
 1&0\\
 1&1\\
 \end{array}
 \right ]
 $$ 
A list of some binary strings and their encodings follows. A proof on the uniqueness of the encoding can be found in \cite{GKP89}. 

\begin{align*}
0 & \hspace{0.1in} 
[\begin{array} {lr}
1 & 2
\end{array}] 
& 00 & \hspace{0.1in}  [\begin{array} {lr}
1& 3
\end{array}] 
& 10 & \hspace{0.1in}  [\begin{array} {lr}
2& 3
\end{array}]  
& 000 & \hspace{0.1in}  [\begin{array} {lr}
1& 4
\end{array}]  & 010 & \hspace{0.1in}  [\begin{array} {lr}
3& 5
\end{array}]   \\
1 & \hspace{0.1in}  [\begin{array} {lr}
2& 1
\end{array}] & 01 & \hspace{0.1in}  [\begin{array} {lr}
3& 2
\end{array}]  & 11 & \hspace{0.1in}  [\begin{array} {lr}
3& 1
\end{array}]  & 001 & \hspace{0.1in}  [\begin{array} {lr}
4& 3
\end{array}]  & 011 & \hspace{0.1in}  [\begin{array} {lr}
5& 2
\end{array}]   
\end{align*}   
 
Given the vector representation $ v_w $ of a string $ w $, it is also possible to decode the string with the following procedure: Let $ |w|=n $ and $ v_w= [\begin{array} {lr}
a& b
\end{array}] $. Set $ w_n=0 $ if $ b>a $, and $ w_n=1 $ otherwise. Subtract the smaller entry from the larger one to obtain $ v_w^{n-1} $ and repeat this routine until you obtain the vector  $ [\begin{array} {lr}
1& 1
\end{array}] $. When the given vector is not a valid representation of a string, then it is not possible to obtain  $ [\begin{array} {lr}
1& 1
\end{array}] $. The  matrices required for this procedure are $ N_0 $, which has the effect of subtracting the value of the first entry of the vector it is multiplied with from the second entry, and $ N_1 $, for the symmetric action. Note that $ N_{0} = M_{0}^{-1}$  and $ N_{1} = M_{1}^{-1} $.

$$ 
N_{0}=
 \left [
 \begin{array}{rr}
 1&-1\\
 0&1\\
 \end{array}
 \right ]~~~~~
N_{1}=
 \left [
 \begin{array}{rr}
 1&0\\
 -1&1\\
 \end{array}
 \right ]
 $$ 
%
%
  
\subsection{Generalized Stern-Brocot Encoding}\label{sec:genstern}

We generalize the scheme mentioned above to strings on alphabets of arbitrary size and present a new method for encoding strings. Let $ \Sigma=\{a_1,a_2,\dots,a_k \} $, and $ w \in \Sigma^*$. With the \textit{generalized Stern-Brocot encoding} method described below, it is possible to uniquely encode $ w $ using a vector of size $ k $ and $ k \times k $ matrices whose entries belong to the set $ \{-1,0,1\} $. Let us note that one can use other methods to encode strings on arbitrary alphabet size using a vector of a smaller dimension but matrices whose entries belong to a larger set.   

 We start with the $ k $ dimensional vector $ [\begin{array} {rrrr}
 1&1&\dots & 1
 \end{array}] $, which represents the empty string. Suppose that  $|w|=n  $. To encode $w$, for $ i=1 $ to $ n $, if $ w_i=a_j $, the vector is multiplied with the  matrix $A_j  $, the $ k $ dimensional identity matrix whose $ j $'th column is replaced with a column of $ 1 $'s. Multiplication with $ A_j $ causes the  $ j $'th entry of the vector to be replaced by the sum of all the entries in the vector. 
 
 Among the different generalizations of the Stern-Brocot fractions, one that appears in \cite{Ga13} under the name of ``Stern's triatomic sequence'' is similar to the encoding we propose for the case $ k=3 $. The similarity lies in the construction of the sequence, but that sequence is not used for the purpose of encoding. As far as we know, no such generalization exists for the case $ k>3 $.
  
 In the following lemma, we prove the uniqueness of this generalized encoding.
 
 \begin{lemma}\label{lem:unique}
 No two distinct strings on $\Sigma$ \textup{ ($ |\Sigma|=k $)} can be represented by the same  vector of size $ k $ using the generalized Stern-Brocot encoding.
 \end{lemma} 
 \begin{proof}
 We will prove by induction on $n$ that if a $k$-dimensional vector $v$ is the generalized Stern-Brocot encoding of a string of length $n$, then $v$ is not the encoding of any other string of length at most $n$. 
  
 The empty string is represented by the $ k $-dimensional vector of 1's.
 The claim clearly holds for $n=0$, since no other strings of at most this length exist.
 Now assume that the claim holds for all natural numbers up to $ n-1 $. Let $ w $ be a string of length $ n $. The vector $ v_w $ representing $ w $ is obtained by multiplying the vector $ v_w^{n-1} $, representing the first $ n-1 $ symbols of $ w $, with $ A_j $ if $ w_n=a_j $. We will examine various possibilities regarding this final multiplication. Note that at a single step, it is possible to modify only a single entry of each vector. Now consider any string $ u \neq w $ with $ |u|=l $ and $ l \leq n $. If $ u $ and $ v $ have the same first $ n-1 $ symbols, then $ v_w^{n-1}=v_u^{l-1} $, the last symbols of the two strings are unequal, and it is not possible to obtain $ v_w=v_u $ since the same vector is multiplied by different matrices. In the remaining case, we know by the induction hypothesis that $ v_w^{n-1}\neq v_u^{l-1} $. If these vectors disagree in more than two entries, there is no way that one can obtain the same vector by multiplying them once with some matrices of the form $A_j$. So we consider the case of the two vectors disagreeing in at most two entries.
 
 Suppose that $ v_w^{n-1}$ and $v_u^{l-1} $ differ only in the  $ i $'th entry. If the final multiplications both work on the $i$'th entries, they will be adding the same number to them, resulting again in vectors differing in their $ i $'th entries. If one or more of the final multiplications deals with another entry, then the final vectors will surely disagree in that entry. It is not possible in any case to end up with equal vectors,
 
 Now suppose that $ v_w^{n-1}$ and $v_u^{l-1} $ differ in two entries.
 If the final multiplications work on the same entry, then the final vectors will disagree in at least one entry.
 In the only remaining case, each one of the vectors is multiplied by a matrix updating a different one of the disagreeing entries. Let us represent the disagreeing entries of the vectors $ v_w^{n-1} $ and $ v_u^{n-1} $ by the pairs $(a, b)$ and $(c, d)$, respectively. Let $ x $ be the sum of the remaining $k-2$ entries in which the vectors agree. Without loss of generality, say that the entries become $(a, a+b+x)$ and $(c+d+x, d)$ after the final multiplication.
 But if the final vectors are equal, these pairs should also be equal, implying  $ c+b+2x=0 $, an impossibility. 
 
 We therefore conclude that it is not possible to have $ v_w=v_u $ for any string $ u $ of length at most $n$.
\end{proof}

Like in the binary case, given the vector representation  of a string, it is possible to reconstruct the string. The all-ones vector corresponds to the empty string. Any other vector $ v_w$ encoding a string $w$ of length $n$ in this encoding has a unique maximum entry, say at position $j$. Then $ w_n$ is $a_j $, and we obtain $ v_w^{n-1} $ by subtracting the sum of the other entries from the greatest entry. One repeats this procedure, reconstructing the string from right to left, until  one ends up with the all-ones vector. In terms of matrices, multiplications with the inverses of $ A_j $'s capture this process.  

\subsection{A Hierarchy Result}  
We will now use the generalized Stern-Brocot encoding to show a hierarchy result based on the dimension of the vector when an additional restriction is imposed on the matrices. 
 
\begin{theorem}\label{thm:hier}
Let $ S $ be the set of matrices whose entries belong to the set $  \{-m,-m+1,\dots,\allowbreak 0,\dots,m-1,m\} $ for some positive integer $m$, and let a \textup{DHVA($k$)} that is restricted to using members of $S$ during its transitions be denoted a \textup{DHVA$_S$($k$)}. Then $ \mathfrak{L}\textup{(DHVA}_S(\textit{k})) \subsetneq \mathfrak{L} \textup{(DHVA}_S(\textit{l}))$ for $ l>(km)^k $. 
\end{theorem}

\begin{proof}
Using the generalized Stern-Brocot encoding, first we will show that it is possible to recognize $ \mathtt{MPAL}_l=\{w\#w^r|w\in\{a_1,a_2,\dots,a_l\}^*\} $ by a DHVA$_S$($ l $) $ \mathcal{V} $.
 
The input alphabet is $ \{a_1,a_2,\dots,a_l\} $,  and the corresponding matrices are $ \{A_1,A_2,\dots,A_l\}, $ described in Section \ref{sec:genstern}. Starting with the $ l $ dimensional vector of 1's,  $ \mathcal{V} $ encodes the string by multiplying its vector with the matrix $ A_j $ whenever it reads an $ a_j $ until it encounters a $ \# $ . After reading the $ \# $, $ \mathcal{V} $ starts decoding by multiplying the vector with matrix $ A_j ^{-1}$ whenever it reads an $ a_j $.

If the string is of the form $ w\# w^r $, the vector will be multiplied with the inverse matrices in the correct order and the resulting value of the vector will be 
$ [\begin{array} {rrr}
 1&1 & \dots 1 
\end{array}]$. 

We also need to show that the input string is not accepted when it is not of the form $ w\#w^r $. Consider an input string $ x\#y^r $ and suppose that it is accepted by $ \mathcal{V} $. Let $ v' $ denote the vector after reading $ x\# $ and let $ Y $ denote the product of the matrices the vector is multiplied while reading $ y^r $. Since the string is accepted, $ v'Y=[\begin{array} {rrr}
1&1& \dots 1 
\end{array}] $ must be true. Since the matrices $ A_j^{-1} $ are invertible, $ Y $ is also invertible, which implies that $ v' $ must be unique. Since $ y\#y^r \in \mathtt{MPAL}$, then $ v' $ must be the vector obtained after reading $ y $ . From Lemma \ref{lem:unique}, we know that every string has a unique representation and we conclude that $ x $ and $ y $ are identical.

We are now going to show that $ \mathtt{MPAL}_l  \notin \mathfrak{L}(\textup{DHVA}_S(k))$ for $ l>(km)^k $. We first note that the value of any entry of a vector of size $ k $ can be at most $ m^{n+1}k^n $ after reading $ n $ symbols. This is possible by letting the initial vector  have $m$ in all entries, and multiplying the vector with the matrix with all entries equal to $ m $  at each step. Similarly, the smallest possible value of an entry is $ -m^{n+1}k^n  $, and so the number of possible different values for a single entry is $ 2m^{n+1}k^n+1 $.  If the machine has $ s $  states, $ s(2m^{n+1}k^n+1)^k $ is an upper bound for the number of different vectors of size $k$ that can be reachable after reading $ n $ symbols. Since there are $ l^n $ strings of length $ n $ when the alphabet consists of $ l $ symbols, for large $ n $  and $ l >(km)^k $, the machine will end up in the same configuration after reading two different strings $ u $ and $ v $. This will cause the strings $ u\#v^r $ and $ v\#u^r $ which are not in  $ \mathtt{MPAL}_{l}$ to be accepted by the machine. Therefore, we conclude that $ \mathtt{MPAL}_{l} \notin \mathfrak{L}(\textup{DHVA}_S(k))$.

Since a vector automaton with a larger vector size can trivially simulate a vector automaton with a smaller vector size, the result follows.
\end{proof}

\section{Relationship with Counter Automata}

We are going to talk about the relationship between homing vector automata and counter automata. A real-time deterministic homing vector automaton with a vector of dimension two can simulate a real-time deterministic one counter automaton (D1CA) which accepts with the condition that the counter is empty (See the proof of Theorem \ref{thm:blind}). The fact that the individual entries of the vector can not be checked prevents us from simulating a real-time deterministic multicounter automaton.

In the following theorem, we show that a DBHVA(2) can recognize a language which is not recognizable by any multicounter machine and we conclude that the language recognition powers of homing vector automata and multi-counter machines are incomparable. 
Note that the result also implies the incomparability of $ \bigcup_k \mathfrak{L} \textup{(DHVA($k$))} $ and $ \bigcup_k\mathfrak{L}\textup{(D\textit{k}CA)} $. This is not the case for the blind versions, as we prove in the second part of the theorem. 

\begin{theorem}\label{th:last}
\begin{enumerate}[i.]
\item $ \bigcup_k \mathfrak{L}\textup{(DBHVA($ k $))} $ and $ \bigcup_k \mathfrak{L} \textup{(D$ k $CA)} $ are incomparable.
\item $ \bigcup_k\mathfrak{L}\textup{(D$ k $BCA)} \subsetneq \bigcup_k \mathfrak{L} \textup{(DBHVA($ k $))}. $
\end{enumerate}
\end{theorem}
\begin{proof}
\label{ap:counter}
 \textit{i.} We know that $ \mathtt{MPAL}_2=\{w\#w^r|w\in\{0,1\}^*\} $ can be recognized by a DBHVA(2) by Theorem \ref{thm:hier}. In \cite{Pe11}, it is proven that no counter machine with $ k $ counters operating in time $ O(2^{n/k}) $ can recognize $ \mathtt{MPAL}_2 $. Since we are working with real-time machines, the result follows.
 
 On the other hand, it is known that the nonregular unary language $ \mathtt{UGAUSS}=\{a^{n^2+n} | n \in \mathbb{N} \}$ can be recognized by a D2CA \cite{SYS13}. By Theorem \ref{thm:unary}, we know that DHVA($ k $)'s and inherently DBHVA($ k $)'s can recognize only regular languages in the unary case. Hence, we conclude that the two models are incomparable.
 
 \noindent \textit{ii.} Let us simulate a given D$ k $BCA $ \mathcal{M} $ by a DBHVA($ k+1 $). Let  $ [\begin{array} {rrrr}
 1&1&\dots & 1
 \end{array}] $
 be the initial vector of $ \mathcal{V} $. $ k+1 $'st entry of the vector will remain unchanged throughout the computation which will allow the counter updates. At each step of the computation, $ \mathcal{V} $ will multiply the vector with the appropriate matrix $ M\in S $ where $ S $ is the set of all $ (k+1)\times (k+1) $ matrices corresponding to possible counter updates. Since each counter can be decremented, incremented or left unchanged, $ |S|=3^k $. All matrices will have the property that $ M(i,i)=1  $ and $ M(k+1,k+1)=1 $. When the $ i $'th counter is incremented and decremented, then  $ M(k+1,i)=1 $ and $ M(k+1,i)=-1 $, respectively. At the end of the computation, the input will be accepted if the vector is equal to $ [\begin{array} {rrrr}
 1&1&\dots & 1
 \end{array}] $, which happens iff all counters have value 0.
   
 The inclusion is proper by the witness language $ \mathtt{MPAL}_2 $.
\end{proof}

We have mentioned that deterministic blind homing vector automaton can recognize the language  $ \mathtt{MPAL}_2 $ which is not recognizable by any counter machine. Consider the language  $ \mathtt{POW}=\{a^nb^{2^n}|n\geq 0\} $, whose Parikh image is not semilinear, which proves that the language is not context-free. Let us note that it is also possible to recognize $ \mathtt{POW} $ by a DBHVA(3) by using the same idea in the proof of Theorem \ref{thm:upow}.  

\section{Open Questions}\label{sec:end}
We focused on real-time computation throughout the paper. What is the power of one-way homing vector automata that are allowed to pause for some steps during their left-to-right traversal of the input string?

Can we show a separation result between the class of languages recognized based on the set of matrices used during the transitions of a homing vector automaton? Most of the homing vector automata we constructed in  the paper are restricted to using matrices whose entries belong to the set $ \{-1,0,1\} $. Is it possible to recognize, for instance, the language $ \mathtt{POW_r}=\{a^{2^n}b^n|n\geq 0\} $ when the matrix entries are restricted to this set? Note that it is possible to construct a DBHVA(2) recognizing $ \mathtt{POW_r} $ with the initial vector 
$  \left [
\begin{array}{rr}
  0&1
    \end{array} 
       \right ] $ 
    and the matrices $$ M_a=\left [
  \begin{array}{rr}
  1&0\\
  1&1\\
  \end{array}
  \right ] \mbox{ and } M_b=\left [
    \begin{array}{rr}
    \frac{1}{2}&0\\
    0&1\\
    \end{array}
    \right ] .$$
    
Can we show a hierarchy result between the class of languages recognized by a homing vector automaton of dimension $ k $ and $ k+1 $ for some $ k>1 $ when the matrix entries are restricted to the set $ \{-1,0,1\} $? Consider the family of languages $ \mathtt{POW}(k)=\{ a^{k^n}b^n|n \geq 0\} $. We conjecture that it is not possible to recognize $ \mathtt{POW}(k)$ with a homing vector automaton of dimension less than $ k+1 $ with the restricted set of matrices.

What can we say about the relationship between homing vector automata and ordinary vector automata? The definition of the vector automaton allows multiplication by a matrix while processing the right end-marker, whereas this is not the case for the homing vector automaton, which makes the comparison between the two models difficult. Would the additional capability of multiplication on the right end-marker increase the computational power of homing vector automata?   

\section*{Acknowledgements}
We thank Ryan O'Donnell and Abuzer Yakary{\i}lmaz for their helpful answers to our questions, and the anonymous reviewers for their constructive comments.

\bibliographystyle{splncs03}
\bibliography{YakaryilmazSay}

\begin{thebibliography}{10}
\providecommand{\url}[1]{\texttt{#1}}
\providecommand{\urlprefix}{URL }

\bibitem{Br61}
Brocot, A.: {Calcul des rouages par approximation, nouvelle m\'{e}thode}. Revue
  Chronom\'{e}trique  3,  186--194 (1861)

\bibitem{Ch62}
Chomsky, N.: Context-free grammars and pushdown storage. M. I. T. Res. Lab.
  Electron. Quart. Prog. Report.  65,  187--194 (1962)

\bibitem{FMR67}
Fischer, P.C., Meyer, A.R., Rosenberg, A.L.: Real time counter machines. In:
  Proceedings of the 8th Annual Symposium on Switching and Automata Theory
  (SWAT 1967). pp. 148--154. FOCS '67 (1967)

\bibitem{FMR68}
Fischer, P.C., Meyer, A.R., Rosenberg, A.L.: Counter machines and counter
  languages. Mathematical Systems Theory  2(3),  265--283 (1968)

\bibitem{Ga13}
Garrity, T.: A multidimensional continued fraction generalization of stern's
  diatomic sequence. Journal of Integer Sequences  16(2), ~3 (2013)

\bibitem{GKP89}
Graham, R., Knuth, D., Patashnik, O.: Concrete Mathematics: \mbox{A} Foundation
  for Computer Science. Addison-Wesley (1989)

\bibitem{Gr78}
Greibach, S.A.: Remarks on blind and partially blind one-way multicounter
  machines. Theoretical Computer Science  7,  311--324 (1978)

\bibitem{ISK76}
Ibarra, O.H., Sahni, S.K., Kim, C.E.: Finite automata with multiplication.
  Theoretical Computer Science  2(3),  271 -- 294 (1976)

\bibitem{Ka09}
Kambites, M.: Formal languages and groups as memory. Communications in Algebra
  37(1),  193--208 (2009)

\bibitem{LR14}
Lipton, R.J., Regan, K.W.: Quantum Algorithms via Linear Algebra. MIT Press
  (2014)

\bibitem{MS97}
Mitrana, V., Stiebe, R.: The accepting power of finite automata over groups.
  In: New Trends in Formal Languages. pp. 39--48. Springer-Verlag (1997)

\bibitem{Pe11}
Petersen, H.: Simulations by time-bounded counter machines. International
  Journal of Foundations of Computer Science  22,  395--409 (2011)

\bibitem{SYS13}
Salehi, {\"O}., Yakary{\i}lmaz, A., Say, A.C.C.: Real-time vector automata. In:
  Proceedings of the 19th International Conference on Fundamentals of
  Computation Theory. pp. 293--304. FCT'13, Springer-Verlag (2013)

\bibitem{St58}
Stern, M.A.: \"{U}ber eine zahlentheoretische \mbox{F}unktion. Journal f\"{u}r
  die Reine und Angewandte Mathematik  55,  193--220 (1858)

\bibitem{Tu69}
Turakainen, P.: Generalized automata and stochastic languages. Proceedings of
  the American Mathematical Society  21,  303--309 (1969)

\end{thebibliography}

\end{document}